\documentclass[journal]{IEEEtran}
\usepackage{calc}

\usepackage{multirow}
\usepackage{cite}
\usepackage{graphicx,subfigure}
\usepackage{psfrag}
\usepackage{amsmath,amssymb}
\usepackage{color}
\usepackage{tikz}
\newenvironment{proof}{\begin{IEEEproof}}{\end{IEEEproof}}

\DeclareMathOperator*{\dotleq}{\overset{.}{\leq}}
\DeclareMathOperator*{\dotgeq}{\overset{.}{\geq}}
\DeclareMathOperator*{\defeq}{\triangleq}

\newtheorem{theorem}{Theorem}
\newtheorem{corollary}{Corollary}[theorem]

\newtheorem{proposition}{Proposition}
\newtheorem{example}{Example} 

\newcommand{\bit}{\begin{itemize}}
\newcommand{\eit}{\end{itemize}}

\newcommand{\bc}{\begin{center}}
\newcommand{\ec}{\end{center}}

\newcommand{\ba}{\begin{array}}
\newcommand{\ea}{\end{array}}

\newcommand{\beq}{\begin{equation}}
\newcommand{\eeq}{\end{equation}}

\newcommand{\beqn}{\begin{equation*}}
\newcommand{\eeqn}{\end{equation*}}

\newcommand{\bean}{\begin{eqnarray*}}
\newcommand{\eean}{\end{eqnarray*}}
\newcommand{\bea}{\begin{eqnarray}}
\newcommand{\eea}{\end{eqnarray}}

\def\E{\mathbb{E}}

\begin{document}
\sloppy

\title{Wireless Coded Caching: A Topological Perspective}

\author{Jingjing Zhang and Petros Elia
\thanks{The authors are with the Mobile Communications Department at EURECOM, Sophia Antipolis, 06410, France (email: jingjing.zhang@eurecom.fr, elia@eurecom.fr).
The work of Petros Elia was supported by the ANR Jeunes Chercheurs project ECOLOGICAL-BITS-AND-FLOPS.}
}

\maketitle

\begin{abstract}
We explore the performance of coded caching in a SISO BC setting where some users have higher link capacities than others. Focusing on a binary and fixed topological model where strong links have a fixed normalized capacity 1, and where weak links have reduced normalized capacity $\tau<1$, we identify --- as a function of the cache size and $\tau$ --- the optimal throughput performance, within a factor of at most 8. The transmission scheme that achieves this performance, employs a simple form of interference enhancement, and exploits the property that weak links attenuate interference, thus allowing for multicasting rates to remain high even when involving weak users.
This approach ameliorates the negative effects of uneven topology in multicasting, now allowing all users to achieve the optimal performance associated to $\tau=1$, even if $\tau$ is approximately as low as $\tau\geq 1-(1-w)^g$ where $g$ is the coded-caching gain, and where $w$ is the fraction of users that are weak. This leads to the interesting conclusion that for coded multicasting, the weak users need not bring down the performance of all users, but on the contrary to a certain extent, the strong users can lift the performance of the weak users without any penalties on their own performance.
Furthermore for smaller ranges of $\tau$, we also see that achieving the near-optimal performance comes with the advantage that the strong users do not suffer any additional delays compared to the case where $\tau = 1$.
\end{abstract}

\section{Introduction}
Recently the seminal work in~\cite{MN14} introduced \emph{coded caching} as a means of using caches at the receivers in order to induce multicasting opportunities that lead to substantial removal of interference. This breakthrough provided impressive throughput gains, and inspired a sequence of other works such as~\cite{MNNU13,MND13,WTP:16,SJTLD:15,CFLsmallCaches:14,SG:16,AG:16,PMN:13,HKS:15,ZLW:15, BWT:16}, \nocite{PK:05ApproxUnivOpti}as well as~\cite{KNMD:14, GSDMC:12, JCM:13,UAS:15,LWG:16,HKD:14}, \nocite{OGKEBooc:12}and even extensions that are specific to wireless networks\cite{SMK:15, MN:15isit,ZFE:15,GKY:15,ZEinterplay:16,NMA:16,WTS:16}.

Focusing on the single-stream broadcast channel, the work in~\cite{MN14} considered a single transmitter with access to a library of $N$ files, serving a set of $K$ receiving users, each requesting a single file from this library. As is typical with caching techniques, the communication was split into two phases: the caching phase and the delivery phase. During the caching phase (off peak hours), each user could cache the equivalent of $M$ files (corresponding to a fraction $\gamma\defeq M/N$ of the library in each cache) without knowledge of what file each user will request. During the delivery phase (peak hours), which would commence upon notification of each user's requested file (one requested file per user), the transmitter would deliver (the remaining of) the single requested file to each user.

Emphasis in~\cite{MN14} was placed on the symmetric, error free, single-stream BC, where each link from the transmitter to any of the receivers was identical, with normalized capacity equal to 1 file per unit of time.
For this topologically symmetric setting, it was shown that a delivery phase with delay $T(K)\defeq \frac{K(1-\gamma)}{1+K\gamma}$ suffices to guarantee the delivery of any $K$ requested files to the users. This was achieved by caching a fraction $\gamma$ of each file at each cache, and then by using cache-aided multicasting to send the remaining information to $K\gamma+1$ users at a time. In this symmetric setting, the resulting coding gain $g_{max}\defeq \frac{K(1-\gamma)}{T(K)} = 1+K\gamma$ far exceeded the local caching gains typically associated to receiver-side caching.

What was also noticed though is that, because of multicasting, the performance suffered when the links had unequal capacities. Such uneven topologies, where some users have weaker channels than others, introduce the problem that any multicast transmission that is meant for at least one weak user, could conceivably have to be sent at a lower rate, thus `slowing down' the rest of the strong users as well.
For example, if we were to naively apply the delivery scheme in \cite{MN14} --- which consisted of a sequential transmission of $\binom{K}{K\gamma+1}$ different XORs (one XOR for each subset of $K\gamma+1$ users) --- we would have the case that even a single weak user would suffice for the performance to deteriorate such that $T(K,\tau) > T(K,\tau=1), \ \forall \tau<1$.
Such topological considerations\footnote{In wireless communications, there is a variety of topological factors --- including propagation path loss, shadow fading and inter-cell interference \cite{TV:05} --- which lead to having some links that are much weaker or stronger than others; a reality that has motivated a variety of works (e.g.~\cite{ETW:08,VKV:11,KV:11,KV:12,KV:12it,GCS:13,HCJ:12}) relating to \emph{generalized} degrees of freedom (GDoF).} have motivated work such as that in~\cite{BWT:16} which --- for the setting of the broadcast erasure channel --- includes a `balancing' solution where only weak users have access to caches, while strong users do not.

Our motivation is to mitigate the performance degradation that coded caching experiences when some link capacities are reduced.
The key to mitigating this topology-induced degradation, is a simple form of interference enhancement which exploits the natural interference attenuation in the direction of the weak links, and which allows us to maintain --- to a certain degree --- a constant multicasting flow of normalized rate 1.
\subsection {Cache-aided SISO BC}

We focus on the topologically-uneven wireless SISO $K$-user broadcast channel, where $K-W$ users have strong links with unit-normalized capacity, while the remaining $W$ users have links that are weak with normalized capacity $\tau$ for some fixed $\tau \in [0,1]$. For notational convenience we will assume that users $1,2,\dots,W$ are weak, and that users $W+1,\dots,K$ are strong.
In this setting, where a single-antenna transmitter communicates to $K$ single-antenna receiving users, at any time $t$, the received signal at user $k$ takes the form
\begin{align}
y_{k,t} = \sqrt{P^{\tau_{k}}}  h_{k,t} x_t + z_{k,t} \  \ \ k =1,2,\cdots,K
\end{align}
where the input signal $x_t$ has bounded power $\E \{|x_t|^2\} \leq 1$, where the fading $h_{k,t}$ and the noise $z_{k,t}$ are assumed to be Gaussian with zero mean and unit variance, and where the link strength is $\tau_k = 1$ for strong users, and $\tau_k = \tau$ for weak users. In this setting, the average received signal to noise ratio (SNR) for the link to user $k$ is given as\footnote{Additionally in the high $P$ regime of interest here, it is easy to see that $Pr(|\sqrt{P^{\tau_{k}}}  h_{k,t} |^2 \doteq P^{\tau_{k}}) =1$.},\footnote{We here use $\doteq$ to denote \emph{exponential equality}, i.e., we write $g(P)\doteq P^{B}$ to denote $\displaystyle\lim_{P\to\infty}\frac{\log g(P)}{\log P}=B$. Similarly $\dotgeq$ and $\dotleq$ will denote exponential inequalities.}
\[\E\{|\sqrt{P^{\tau_{k}}}  h_{k,t} x_t|^2\}  = P^{\tau_{k}}. \]

\begin{figure}[t!]
  \centering
\includegraphics[width=0.8\columnwidth]{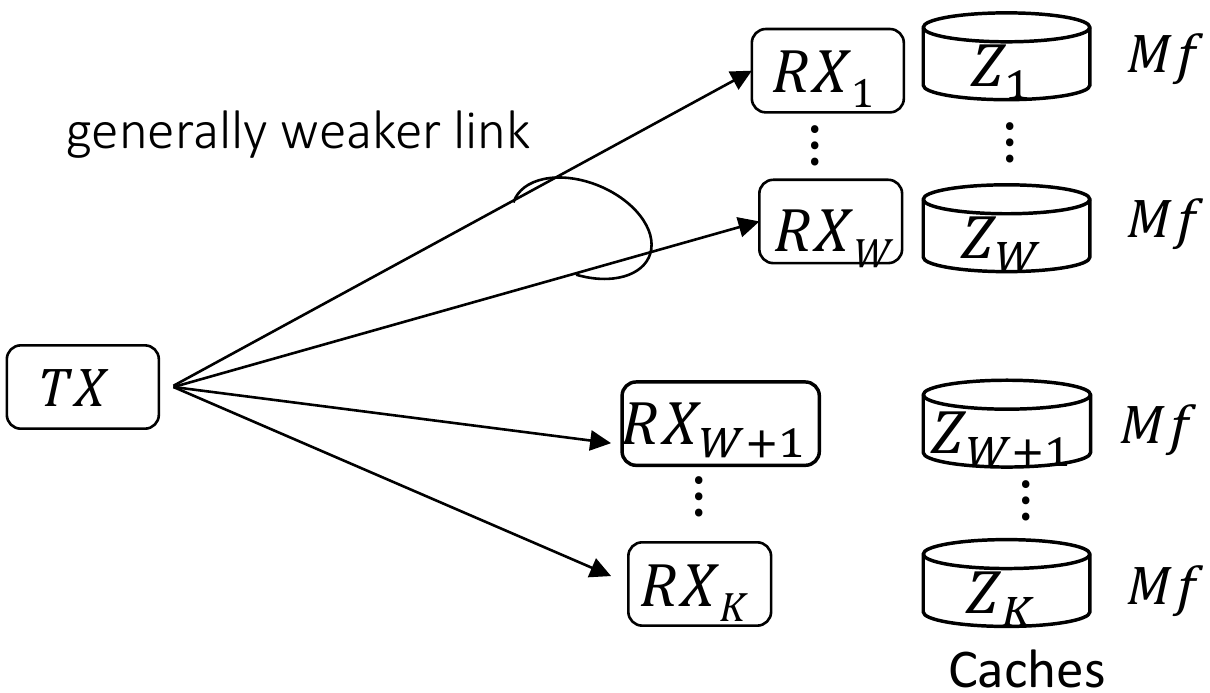}
\caption{Cache-aided $K$-user MISO BC.}
\label{fig:model}
\end{figure}

We focus on the high SNR (high $P$) setting, and we make the normalization --- without loss of generality --- that each library file $W_n, \ n=1,\dots,N$, has size $f$ (bits) which --- in the high SNR setting of interest here --- is set equal to $f = \log_2(P)$. Consequently the aforementioned capacity of a strong (interference free) link, is now \emph{1 file per unit of time}, while the capacity of a weak link is $\tau$ files per unit of time. The cache $Z_k$ of user $k$ has size $Mf$ bits, where $M$ ($M \leq N$) defines the aforementioned normalized cache size
\begin{align} \label{eq:gamma1}
\gamma \defeq \frac{M}{N}.
\end{align}
Our results consider the case where $N \geq K$, and consider the measure of performance $T$ --- in time slots, per file served per user --- needed to complete the delivery process, \emph{for any request}. After the aforementioned normalization $f = \log (P)$, this measure matches that in \cite{MN14}.

\subsection{Notation and conventions}
We will use $\mathcal{K} \defeq \{1,2,\cdots,K\}$ to denote the (indices of the) set of all users, $\mathcal{W} \defeq\{1,2,\cdots,W\}$ to denote the set of weak users, and $\mathcal{S} \defeq \{W+1,\cdots,K\}$ to denote the set of strong users. We will also use $w\defeq W/K$ to define the fraction of the users that are weak. We remind the reader that $\binom{n}{k}$ will be the $n$-choose-$k$ operator, and $\oplus$ will be the bitwise XOR operation. If $A$ and $B$ are two sets, then $A \backslash B$ denotes the difference set. For a transmitted signal $x$, we will use $\text{dur}(x)$ to denote the transmission duration (in units of time) of that signal. 
We will use $\Gamma \defeq \frac{KM}{N} = K\gamma$ to denote the cumulative (normalized) cache size, and for any integer $L$, we will use \begin{align}\label{eq:Tcentral}T(L)\defeq \frac{L(1-\gamma)}{1+L\gamma}\end{align} to denote the delay associated to the original coded caching solution in \cite{MN14} with $L$ strong users and no weak users ($\tau = 1$).

Consequently we will use $T(K) \defeq \frac{K(1-\gamma)}{1+K\gamma}$ to describe the performance for the case of $L=K$ users, as this was derived
in~\cite{MN14} for integer $K\gamma= \{0,1,\cdots,K\}$ (for the general $K\gamma$, the lower convex envelope of the integer points is achievable).
Similarly $T(K-W)=\frac{(K-W)(1-\gamma)}{1+(K-W)\gamma}$ will simply correspond to the case of $L=K-W$, and $T(W)=\frac{W(1-\gamma)}{1+W\gamma}$ to the case of $L=W$, and \emph{we stress that $T(K),T(K-W),T(W)$ all correspond to the case of $\tau = 1$}. We here note that for clarity of exposition, we allow for an integer relaxation on $(K-W)\gamma$ and $W\gamma$. This relaxation, which allows for crisp expressions, will be lifted in Section~\ref{sec:MNextension} which, for completeness, presents the extension of the algorithm in~\cite{MN14} for any $\gamma$, using memory-sharing between files (see also\cite{HA:2015}).

\section{Throughput of topological cache-aided BC}
The following describes, within a factor of 8, the optimal $T(\tau)$ as a function of $K,W,\gamma,\tau$. The results use the expression
\[\bar \tau_{thr} = \frac{T(W)}{T(W)+T(K-W)}\]
and
 \begin{align} \label{eq:tauThres}
    \tau_{thr} = \left\{\begin{array}{lr}
        1-\frac{\binom{K-W}{K \gamma+1}}{\binom{K}{K \gamma+1}}, & \text{for } W < K(1-\gamma)\\
        1, & \text{otherwise}.
        \end{array}\right.
 \end{align}

 The following applies to the case of centralized placement.
\vspace{3pt}
\begin{theorem} \label{thm:centralT}
In the $K$-user topological cache-aided SISO BC with $W$ weak users,
\begin{align}
T(\tau) \!= \!\left\{ \!\!\!{\begin{array}{*{20}{c}}
\frac{T(W)}{\tau}, & \! 0 \leq \tau  < \bar\tau_{thr} \\
\min \{T(K-W)+T(W), \frac{\tau_{thr}  T(K)  }{\tau}\}, & \!\bar\tau_{thr} \leq  \tau \leq \tau_{thr} \\
T(K), &  \!\tau_{thr}  <  \tau \leq 1
\end{array}} \right.
\end{align}
is achievable, and has a gap from optimal
\begin{align}
\frac{T(\tau)}{T^*(\tau)} \leq 8
\end{align}
that is always less than 8.
\end{theorem}
\vspace{3pt}

\vspace{3pt}
\begin{proof}
The scheme that achieves the above performance is presented in Section~\ref{sec:achievability}, while the corresponding gap to optimal is bounded in Appendix~\ref{sec:gapproof}.
\end{proof}
\vspace{3pt}

What the above shows is that there are three regions of interest. In the first region where $\tau\geq \tau_{thr}$, despite the degradation in the link strengths, the performance of all users remains as if all links were uniformly strong (as if $\tau = 1$). In this setting, instead of experiencing the phenomenon that the weak users `pull down' the performance of all users, we observe the interesting effect of strong users bringing up the performance of the weak users, to the optimal $T(K)$ associated to $\tau = 1$. The conclusion is that in this first region, the reduction in the capacity of the weak links $\tau$, does not translate into a performance degradation. This is because, even when multicasting involves weak users, the employed superposition scheme allows for an overall multicasting rate of 1.
Then, there is an intermediate region where there is a degradation in the overall performance by a factor $\frac{\tau_{thr}}{\tau}$ (rather than by a factor $\frac{1}{\tau}$). Finally there is the third region $\tau \leq \bar{\tau}_{thr}$, where due to the substantially limited capacity of the weak links, the transmission to the weak users becomes the bottleneck and the performance is dominated by the delay of serving the weak users, and it deteriorates by a factor $\frac{1}{\tau}$.
Interestingly, within this region, and particularly when $\tau \in [0, \frac{w}{1+K(1-w)\gamma}]$, while the near optimal performance reflects the bottleneck due to the weak users, it is also the case (this can be seen in the description of the scheme) that the delivery to the strong users finishes much earlier, and that the strong users do not suffer any additional delays compared to the case where $\tau = 1$; each strong user completes reception of their file with delay that is not bigger than $T(K)$. 

In all cases, we see an improvement over the aforementioned naive sequential transmission of XORs, for which it is easy to show that the performance takes the form
\begin{align}
T_{nv} & = T(K)\bigl( 1+\frac{\tau_{thr}}{\tau}(1-\tau)  \bigr) \\ &= \frac{T(K)}{\tau}\bigl(  1- (1-\tau_{thr})(1-\tau) \bigr)
\end{align}
where we see that $T_{nv}(\tau) >T(K)$ for any $\tau<1$.
The gains of the proposed method, compared to the naive sequential multicasting, are more prominent when $\tau$ is reduced ($0 \leq \tau  < \bar\tau_{thr}$), and when $K\gamma>1$ and $W\gamma<1$, in which case the gains are bounded as 
\[\frac{T_{nv}}{T(\tau)} < \frac{2}{W\gamma}\]
and can become large when $W\gamma$ becomes substantially small.
\vspace{3pt}

\begin{example} ($K = 500, W = 50,\gamma = \frac{1}{50}$)
Directly from the above we see that
\begin{align}
T =\left\{ {\begin{array}{*{20}{c}}
\frac{24.5}{\tau}, & 0 \leq \tau < 0.36  \\
\min \{68.6, \frac{30.7}{\tau}\}, & 0.36 \leq  \tau \leq 0.69 \\
T(K) = 44.5, &  0.69 <  \tau \leq 1
\end{array}} \right.
\end{align}
which means that, with a tenth of the users being weak, as long as $\tau\geq 0.69$, there is no performance degradation due to reduced-capacity links, and every user receives their file with delay $T(K) = \frac{K(1-\gamma)}{1+K\gamma} = 44.5$ associated to $\tau = 1$.
\end{example}%

\vspace{3pt}
Regarding the first region, the following quantifies the intuition that the topology threshold $\tau_{thr}$ (until which, capacity reductions do not degrade performance), is a function of the degree of multicasting (coding gain) $g_{max}\defeq K \gamma +1 = K(1-\gamma)/T(K)$.

\begin{corollary} \label{cor:tauthr}
The threshold $\tau_{thr}$ which guarantees full-capacity performance $T(K)$, lies inside the region $\tau_{thr}\in[1-(1-w)^{g_{max}},1-(1-w-\frac{w\gamma}{1-\gamma})^{g_{max}}]$, which also means that \[T(\tau) = T(K), \ \forall \tau\geq 1-(1-w)^{g_{max}} + \gamma^{g_{max}}.\]
Thus as $\gamma$ decreases, this threshold approaches 
\[\tau_{thr}\approx 1-(1-w)^{g_{max}}.\]
\end{corollary}
\begin{proof}
The proof consists of basic algebraic manipulations and can be found in the Appendix.
\end{proof}

We again note that a simple sequential delivery of the XORs would have resulted in $\tau_{thr}=1$.

We extend the above to the link-capacity threshold
\begin{align}
\tau_{thr,G} \defeq\arg\min \{\tau: T(\tau) \leq G \cdot T(K), \ G \geq 2 \label{eq:definitionth}\}
\end{align}
until which, the performance loss is restricted to a factor of $G\geq 2$. For example, for any $\tau\geq \tau_{thr,2}$, the scheme guarantees that $T(\tau) \leq 2T(K)$.

\begin{figure}[t!]
  \centering
 \includegraphics[scale=0.35]{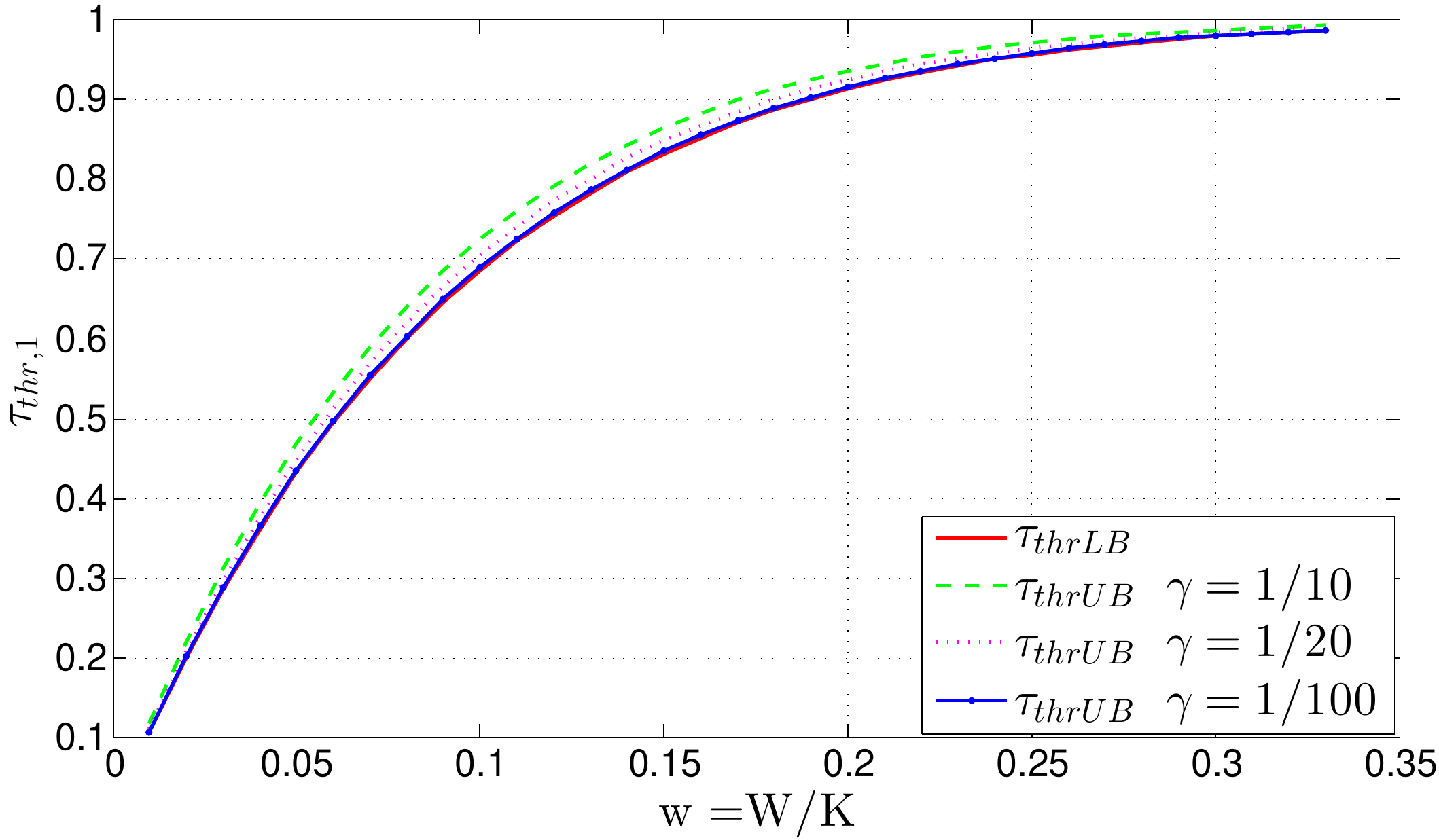}
\caption{$\tau_{thrLB} =1-(1-w)^{g_{max}}$ denotes the lower bound of $\tau_{thr}$, while $\tau_{thrUB} = 1-(1-w-\frac{w\gamma}{1-\gamma})^{g_{max}}$ denotes the upper bound.}
\label{tauThr1Bound}
\end{figure}

\begin{corollary} \label{cor:thresholdG}
For any $\tau \geq \tau_{thr,G} = \frac{w}{1+w(g_{max}-1)} \frac{g_{max}}{G}$ ($G \geq 2$), the performance degradation is bounded as $T(\tau) \leq G \cdot T(K).$
\end{corollary}
\vspace{3pt}
\begin{proof}
The proof is presented in Appendix~\ref{sec:threshold}.
\end{proof}
\vspace{3pt}

\begin{example} ($w = \frac{1}{10}, g_{max}=11$)
Here, as we have seen, $\tau_{thr} = 0.686$, whereas
\begin{align}
\tau_{thr,G} = \frac{0.55}{G},  \ G \geq 2
\end{align}
which means that any link-capacity reduction down to, for example, $\tau\geq \tau_{thr,2} = \frac{0.55}{2} = 0.275$, only comes with a performance deterioration of at most 2 ($T(\tau)\leq 2 T(K), \ \forall \tau\geq 0.275$).
\end{example}

\subsection{Decentralized case}
We proceed to provide similar results for the case of decentralized placement, where as described in~\cite{MND13}, the caching phase is a random process. The result takes the same form as above, except that now we substitute $T(L)$ from \eqref{eq:Tcentral} with the decentralized equivalent $T'(L) = \frac{1-\gamma}{\gamma} (1-(1-\gamma)^L)$ ($L=K,K-W,W$) (cf.~\cite{MND13}), and where we substitute $\tau_{thr},\bar\tau_{thr}$ with
\begin{align}
\tau'_{thr} &= \frac{1-(1-\gamma)^W}{1-(1-\gamma)^K}, \ \bar\tau'_{thr} =\frac{1-(1-\gamma)^W}{2-(1-\gamma)^W-(1-\gamma)^{K-W}} \notag .
\end{align}
For completeness we present the result below.
\vspace{3pt}
\begin{theorem} \label{theoremde}
In the $K$-user topological cache-aided SISO BC with $W$ weak users, and decentralized cache placement,
\begin{align}
T =\left\{ {\begin{array}{*{20}{c}}
\frac{T'(W)}{\tau} , & 0 \leq \tau < \bar\tau'_{thr} \\
T'(K-W) + T'(W), & \bar\tau'_{thr} \leq \tau \leq \tau'_{thr} \\
T'(K), &  \tau'_{thr}  <  \tau \leq 1
\end{array}} \right.
\end{align}
is achievable and order optimal. 
\end{theorem}
\vspace{3pt}

\begin{figure}[t!] \label{tauThr1gmax}
  \centering
 \includegraphics[scale=0.35]{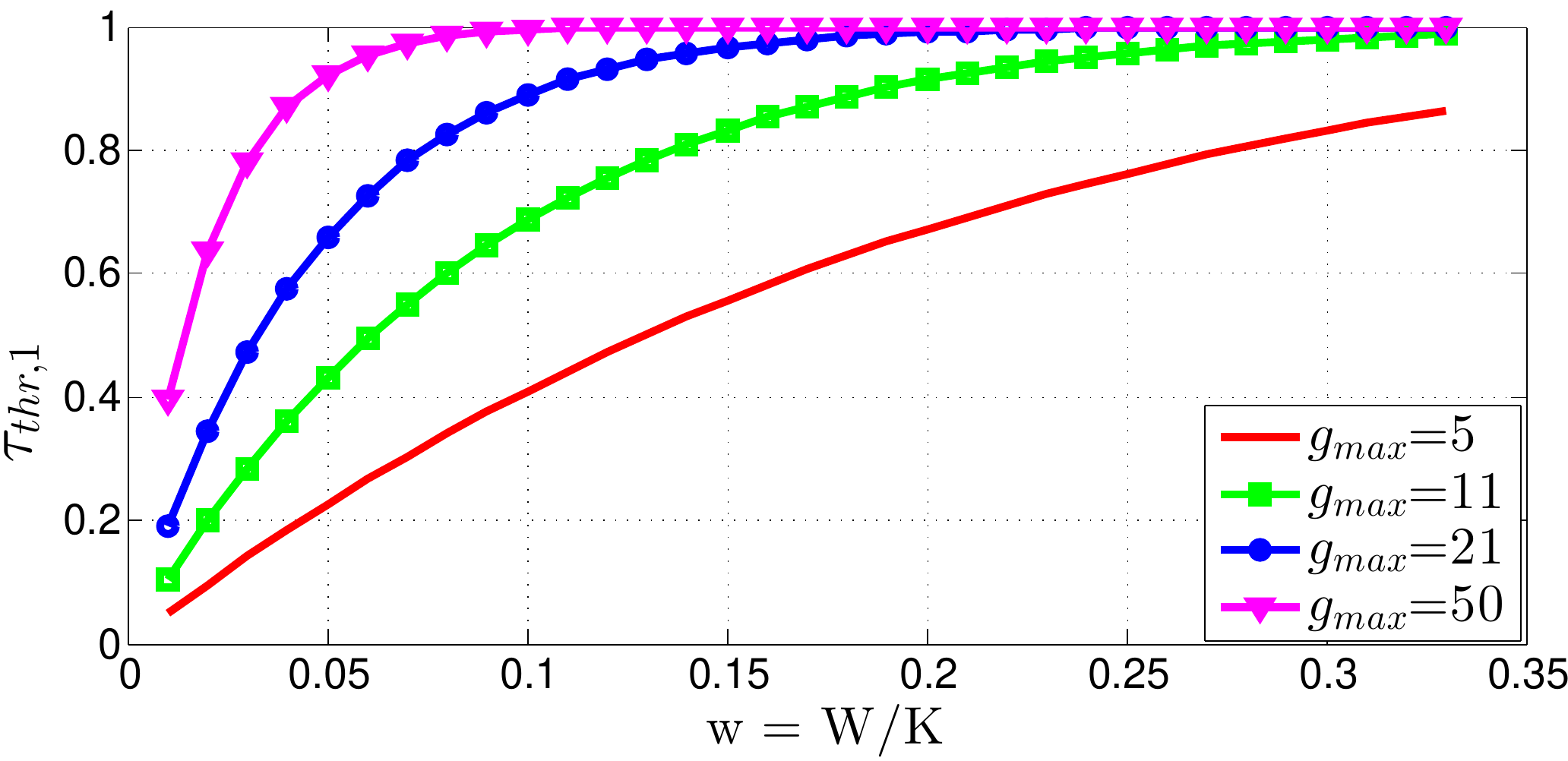}
\caption{$\tau_{thr}$ corresponding to distinct values for gains $g_{max}$. For example, for $g_{max}=5$ and $w=0.1$ then $\tau_{thr}\approx 0.4$.}
\label{fig:vanishingDelayedCSIT2}
\end{figure}

The delivery scheme that allows for the above, is identical to the one in the centralized setting (see below), and the only difference is in the analysis of $T(\tau)$ which accounts for the new thresholds $\tau'_{thr},\bar\tau'_{thr}$. The claim that the scheme is order optimal, follows from the arguments in \cite{MND13} and the arguments in the proof of the gap in the previous theorem. 

\section{Coded caching with simple interference enhancement} \label{sec:achievability}

We proceed to describe the scheme, for the cases in Theorem~\ref{thm:centralT}.

\subsection{Scheme for $\tau \geq \tau_{thr}$}
The following applies to the case where $W < K(1-\gamma)$.

\begin{figure}[t!]
  \centering
 \includegraphics[scale=0.52]{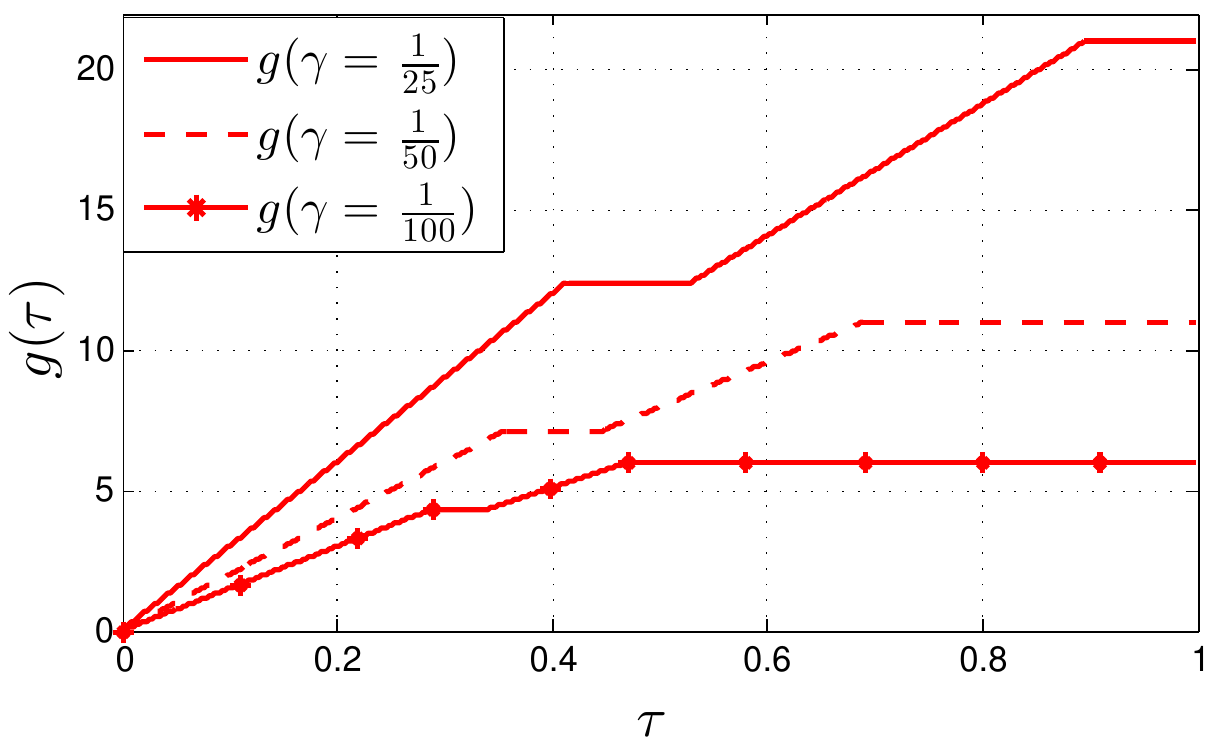}
\caption{The plot shows the gain as a function of $\tau$ when $K =500, W=50$. The horizontal lines represent the maximum gain $g_{max}$ corresponding to $\tau=1$, and demonstrate how these can be achieved even with lesser link capacities.}
\label{tauThr1Bound}
\end{figure}

\subsubsection{Placement phase}
The placement phase is identical to that in~\cite{MN14}, where we recall that each file $W_n, \ n=1,\dots,N$ is equally split into
$\binom{K}{\Gamma}$ subfiles $\{W_{n,\tau}\}_{\tau \in \Psi_{\Gamma}}$ 
where $\Psi_{\Gamma}\defeq \{\tau \subset \mathcal{K} \ : \  |\tau| = \Gamma\}$, such that each cache $Z_k$ is then filled according to 
$Z_k=\{W_{n,\tau}\}_{n \in [N], \tau\in \Psi_{\Gamma},k \in \tau}$.

\subsubsection{Delivery phase}

At the beginning of the delivery phase, the transmitter must deliver each requested file $W_{R_k}$ to each receiver $k$, by delivering the remaining (uncached) subfiles $\{W_{R_k,\tau}\}_{k \notin \tau}$ for each user.

We first recall from~\cite{MN14} that for any $\psi \in  \Psi_{\Gamma+1} \defeq \{\psi\in \mathcal{K} \ : ~ |\psi|=\Gamma+1  \}$, then
\begin{align}
X_{\psi} \defeq \oplus_{k \in \psi} W_{R_k,\psi \backslash \{k\}}
\end{align}
suffices to deliver to each user $k\in \psi$, their requested file $W_{R_k,\psi \backslash k}$.
To satisfy all requests $\{W_{R_k} \backslash Z_k \}_{k=1}^K$, the entire set $\mathcal{X}_\Psi \defeq \{ X_{\psi}\}_{\psi \in  \Psi_{\Gamma+1}}$
consisting of $|\mathcal{X}_\Psi|=\binom{K}{\Gamma+1}$
folded messages (XORs), must be delivered, where each XOR contains (has size)
\begin{align} \label{eq:XpsiSize}
|X_{\psi}| & = |W_{R_k,\tau}| = \frac{f}{\binom{K}{\Gamma}} \ \text{(bits)}.
\end{align}
We distinguish between the subset of XORs $\mathcal{X}_{\Psi,s} \defeq \{ X_{\psi}: \forall \psi, \ s.t. \ \psi \cap \mathcal{W} = \emptyset\} \subset \mathcal{X}_\Psi$ that are only intended for strong users, and the remaining subset $\mathcal{X}_{\Psi,w} \defeq \mathcal{X}_\Psi \backslash \mathcal{X}_{\Psi,s}$ that have at least one weak user as an intended recipient.


Let $T_1$ be the duration required to deliver all of $\mathcal{X}_{\Psi,w}$, to all weak users $k\in \mathcal{W}$.
Let the transmission first take the form
\begin{align} \label{eq:transmissionform}
x_t = c_t + b_t, \  \ \ t \in [0, T_1]
\end{align}
where the power and rate of the symbols are allocated such that
\begin{align} \label{eq:PowerAllocation}
\E\{|c_t|^2\} &\doteq  P^0 , \ \   r_t^{(c)} = \tau  \\
\E\{|b_t|^2\} &\doteq P^{-\tau},  r_t^{(b)} = 1-\tau
\end{align}
where $r_t^{(c)}$ (resp. $r_t^{(b)}$) denotes the prelog factor of the number of bits $r_t^{(c)} f$ carried by symbol $c_t$ (resp. $r_t^{(b)}$) at time $t$. In the above, $c_t$ will carry information from $\mathcal{X}_{\Psi,w}$, while $b_t$ will carry the information from $\mathcal{X}_{\Psi,s}$. As we see, the reduced power of $b_t$ guarantees that it does not interfere with weak users (at least not above the noise level).

During this period, the received signals $y_{k,t}$ take the form
\begin{align}
y_{k,t} &= \underbrace{\sqrt{P} h_{k,t} c_t}_{P} + \underbrace{\sqrt{P} h_{k,t} b_t}_{P^{1-\tau}} + \underbrace{z_{k,t}}_{P^0}, \ \ k \in  \mathcal{S} \\
y_{k,t} &=\underbrace{\sqrt{P^{\tau}} h_{k,t} c_t}_{P^\tau} + \underbrace{\sqrt{P^{\tau}} h_{k,t} b_t}_{P^0} +\underbrace{z_{k,t}}_{P^0}, \ \ k \in \mathcal{W}
\end{align}
allowing each weak user to directly decode $c_t$, and allowing each strong user $ k \in  \mathcal{S}$ to first decode $c_t$ by treating $b_t$ as noise, and to then decode $b_t$ by removing $c_t$. This is achieved because the interference to the strong users was enhanced (see \cite{DavoodiJafarInterfEnhanc:15} and \cite{MDGinterfEnhanc:12}) in order for it to be removed.

Depending on the size of $\mathcal{X}_{\Psi,w}$ and $\mathcal{X}_{\Psi,s}$, we will have two cases. In the first case, all the information in $\mathcal{X}_{\Psi,s}$ is delivered by $b_t$ within the aforementioned duration $T_1$, and thus $T=T_1$. In the second case though, the delivery of $\mathcal{X}_{\Psi,s}$ takes longer than the delivery of $\mathcal{X}_{\Psi,w}$ (longer than $T_1$), in which case the remaining information is transmitted during an additional period of duration $T_2$, during which the transmission (as it is intended only for strong users) takes the simpler form
\begin{align}\label{eq:txSecondCase}
x_t = c_t, \  \ \ t \in [T_1, T_1 + T_2]
\end{align}
during which the power and rate are set as
\begin{align}\label{eq:PowerAndRateSecondCase}
\E\{|c_t|^2\} &\doteq  P^0 , \ \   r_t^{(c)} = 1
\end{align}
which allows each strong user to directly decode $c_t$.

In both cases, each strong user can decode $\mathcal{X}_{\Psi,w}$ and $\mathcal{X}_{\Psi,s}$, while each weak user can decode $\mathcal{X}_{\Psi,w}$, and the delivery process is completed. 

\subsubsection{Calculation of $T$ \label{sec:calcT}}

To calculate the duration of the delivery phase, let us use
\[Q_{\bar{w}} \defeq  |\mathcal{X}_{\Psi,s}| |X_{\psi}| = \frac{\binom{K-W}{\Gamma+1} f}{\binom{K}{\Gamma}} \ \ \text{(bits)} \]
to denote the size (in bits) of $\mathcal{X}_{\Psi,s}$, and let us use
\[Q_{w} = |\mathcal{X}_{\Psi}||X_{\psi}| - Q_{\bar{w}}\ \ \text{(bits)} \]
to denote the size of $\mathcal{X}_{\Psi,w}$. We now treat the aforementioned two cases.
\paragraph{Case 1a: $T_1 > \frac{Q_{\bar{w}}}{(1-\tau)f}$ (this corresponds to $\tau\in[0 , \tau_{thr}]$)} \label{case:moderate}
Here $T=T_1$ is directly calculated, and takes the form
\begin{align} \label{eq:moderateCase}
T = T_1 = \frac{Q_{w}}{\tau f}  =  \frac{1}{\tau} \big(1-\frac{\binom{K-W}{\Gamma+1}}{\binom{K}{\Gamma+1}}\big) \frac{K(1-\gamma)}{1+K\gamma}= \frac{\tau_{thr}T(K)}{\tau}.
\end{align}
\paragraph{Case 1b: $T_1 \leq \frac{Q_{\bar{w}}}{(1-\tau)f}$ (this corresponds to $\tau\in(\tau_{thr},1]$)}
The transition to this new case, happens as soon as $T_1 < \frac{Q_{\bar{w}}}{(1-\tau)f}$, which happens as soon as $\tau > \tau_{thr}$ (i.e., $\tau = \tau_{thr}$ is derived by setting $T_1 = \frac{Q_{\bar{w}}}{(1-\tau)f}$).
Recall that now $T=T_1+T_2$. We can easily calculate that the second period (during which we multicast to strong users at full rate) has duration
\[T_2 = \frac{Q_{\bar{w}}- (1-\tau)f T_1}{f}\]
where $Q_{\bar{w}}-(1-\tau)f T_1$ is the amount of the remaining information of $\mathcal{X}_{\Psi,s}$ that had not been handled during the first period of duration $T_1$. Adding the two components gives us
\begin{align}
T = T_1+T_2 = \frac{K(1-\gamma)}{1+K\gamma}=T(K)
\end{align}
which matches the aforementioned performance $T(K)$ corresponding to uniformly strong topology ($\tau=1$).

\subsection{Scheme for the case of $\tau \leq \tau_{thr}$}
The following applies for all $W\leq K$. Here the idea is that, because the weak link capacities are small, we treat the weak users separately from the strong users. While we generally transmit to both strong and weak users simultaneously, caching at the strong users is independent of the caching at the weak users, and each XOR is meant either for strong users exclusively, or for weak users exclusively. Transmission again takes the form $x_t = c_t+b_t$, and $c_t$ will deliver the group of XORs meant for weak users, while $b_t$ will deliver the group of XORs for the strong users.

For the case of the weak users, the total information that will be sent is $f T(W) \log (P) $ bits, while for the strong users, this will be $fT(K-W) \log (P) $ bits. There will be again two cases, where the split is again a function of the amount of information that needs to be delivered to the weak vs. to the strong users. In the first case, the transmission and allocation of power and rate, are the same as in~\eqref{eq:transmissionform} and~\eqref{eq:PowerAllocation}, while in the second case they will be the same as in~\eqref{eq:txSecondCase} and \eqref{eq:PowerAndRateSecondCase}.
\paragraph{Case 2a: $\frac{f T(K-W)}{(1-\tau)f} < \frac{f T(W)}{\tau f}$ (corresponds to $\tau \in [\bar{\tau}_{thr} ,\tau_{thr}]$)}
For this case --- corresponding to the scenario where the delivery to the strong users does not take longer than the delivery to the weak users --- $T$ can be readily calculated to be
\[T =  \frac{f T(W)}{\tau f} = \frac{T(W)}{\tau}.\]
\paragraph{Case 2b: $\frac{f T(K-W)}{(1-\tau)f} \geq \frac{f T(W)}{\tau f}$, (corresponds to $\tau\in[0,\bar{\tau}_{thr}]$)}
In this second case, in addition to the above mentioned $T_1 = \frac{T(W)}{\tau}$, the second period duration $T_2$ is readily calculated to be \[T_2 = \frac{f T(K-W)- (1-\tau)f T_1}{f}\]
which eventually gives
\begin{align}
T = T_1+T_2 = T(K-W)+T(W).
\end{align}
Combining this with the results corresponding to cases 1a and 2b, gives the desired
\[T(\tau)= \min \{T(K-W)+T(W), \frac{\tau_{thr}  T(K)  }{\tau}\}.\]

\section{Conclusion}
In this work we explored the behavior of coded caching in the topological broadcast channel (BC), identifying the optimal cache-aided performance within a multiplicative factor of 8. Our proposed scheme uses a simple form of interference enhancement to alleviate the negative effect of having to multicast to both strong and weak links. By showing that the optimal performance can be achieved even in the presence of weaker links, the work reveals a new role of coded caching which is to partially balance the performance between weaker and stronger users, and to a certain degree without any penalty to the performance of the stronger users.

\section{appendix}

\subsection{Proving the gap to optimal} \label{sec:gapproof}
To prove the gap to optimal in Theorem~\ref{thm:centralT}, we first recall from~\cite{HA:2015} (which corresponds to the case of $\tau=1$) that $\frac{T(K)}{T^*(\tau=1)} \leq 4$. 
Let us consider the following three cases.

\paragraph{Case 1 ($\tau_{thr}  <  \tau \leq 1$)} In this case, the bound is direct, by seeing the following
\[\frac{T(\tau)}{T^*(\tau)} = \frac{T(K)}{T^*(\tau)} \leq \frac{T(K)}{T^*(\tau=1)} \leq 4.\]

\paragraph{Case 2 ($\bar\tau_{thr} \leq  \tau \leq \tau_{thr}$)}  We first recall that $T(K)$ is increasing with $K$, since
\[T(K)=\frac{K(1-\gamma)}{1+K\gamma} = \frac{1-\gamma}{\gamma} (1-\frac{1}{1+K\gamma}).\]
This means that $T(K-W) \leq T(K)$ and $T(W) \leq T(K)$, and consequently that
\begin{align} \label{eq:gap2}
T(\tau) &=\min \{T(K-W)+T(W), \frac{\tau_{thr}  T(K)  }{\tau}\} \\
  &\leq T(K-W)+T(W) \leq 2T(K)
\end{align}
which yields the desired
\[\frac{T(\tau)}{T^*(\tau)} \leq \frac{2T(K)}{T^*(\tau)} \leq \frac{2T(K)}{T^*(\tau=1)} \leq 8.\]

\paragraph{Case 3 ($0  <  \tau \leq \bar\tau_{thr}$)} For this case, to get a lower bound on $T(\tau)$, we use the bound in~\cite{HA:2015} for a system with $K=W$ users, all of them having a link of capacity $\tau$. This means that the lower bound in \cite{HA:2015} holds, after simple normalization (division) by $\tau$. At the same time, we know that for this case, the achievable performance here is $\frac{T(W)}{\tau}$. Given that the normalization of the lower bound, matches the normalization of the achievable performance, then the gap remains, as in~\cite{HA:2015}, equal to $\frac{T}{T^*} \leq 4 .$

Combining the above three cases, yields the desired
\[\frac{T}{T^*}  \leq 8\]
which completes the proof.

\subsection{Proof of Corollary~\ref{cor:tauthr}} \label{sec:threshold0}
From \eqref{eq:tauThres} we recall that for $W < K(1-\gamma)$ then $\tau_{thr} =  1-\frac{\binom{K-W}{K \gamma+1}}{\binom{K}{K \gamma+1}}$. To simplify we note that
\begin{align}
\frac{\binom{K-W}{K \gamma+1}}{\binom{K}{K \gamma+1}} &= \frac{K-W}{K} \frac{K-W-1}{K-1} \cdots \frac{K-W-K\gamma}{K-K\gamma} \notag \\
                                                      &= (1-w) (1-w-\frac{w}{K-1}) \cdots (1-w-\frac{w K\gamma}{K-K\gamma}) \notag \\
																											&= \prod^{K\gamma}_{i=0} (1-w-\frac{wi}{K-i}) \label{eq:bound}
\end{align}
where the first equation comes from expanding the binomial coefficients $\binom{K-W}{K \gamma+1}$ and $\binom{K}{K \gamma+1}$. Since $\frac{wi}{K-i}$ is increasing with $i$, we have $0 \leq \frac{wi}{K-i} \leq \frac{wK\gamma}{K-K\gamma}$. Applying this inequality to the last equation above (cf.\eqref{eq:bound}), gives
\[(1-w-\frac{w\gamma}{1-\gamma})^{g_{max}} \leq \frac{\binom{K-W}{K \gamma+1}}{\binom{K}{K \gamma+1}} \leq (1-w)^{g_{max}} \]
which in turn gives the lower and upper bound of $\tau_{thr}$, in the form $\tau_{thrLB}=1-(1-w)^{g_{max}}$ and $\tau_{thrUB} = 1-(1-w-\frac{w\gamma}{1-\gamma})^{g_{max}}$.
It is easy to show that the difference between the upper and lower bound is not larger than $\gamma^{K\gamma+1}$, which vanishes as $\gamma$ decreases.%
%
%

\subsection{Proof of Corollary~\ref{cor:thresholdG}} \label{sec:threshold}

Let us recall from \eqref{eq:gap2} that when $\bar\tau_{thr} \leq  \tau \leq \tau_{thr}$ then
\begin{align}
T(\tau) &=\min \{T(K-W)+T(W), \frac{\tau_{thr}  T(K)  }{\tau}\} \\
  &\leq T(K-W)+T(W) \leq 2T(K)
\end{align}
which, together with the fact that $G\geq 2$, implies that such a performance degradation (beyond a factor of 2), requires that $\tau < \bar \tau_{thr}$, which in turn says that the achievable $T(\tau)$ takes the form $T(\tau) = \frac{T(W)}{\tau}$. Applying this in the definition in~\eqref{eq:definitionth}, yields the presented $\tau_{thr,G}$.

\subsection{Removing the integer relaxation constraint} \label{sec:MNextension}

To remove the aforementioned integer relaxation, we consider the extension of the centralized MN algorithm in~\cite{MN14}, to any value of $\gamma$ (not just when $K\gamma$ is an integer). This has already been addressed in \cite{HA:2015} which plots the intermediate values. For the sake of completeness we proceed to explicitly describe the corresponding performance, achieved here by the memory-sharing scheme described below. The following holds for any $\gamma$ and for $\tau= 1$.
\vspace{3pt}
\begin{proposition}\label{prop:generalT}
In the $K$-user cache-aided SISO BC, with $N\geq K$ files and cache size such that $K\gamma \in [t,t+1], t=0,1,\cdots, K-1$, then
\begin{align}
T''(K) &= \big((t+1)-K\gamma \big) \frac{K-t}{t+1} + (K\gamma-t)\frac{K-(t+1)}{t+2} \notag \\
       &= \frac{K-t}{t+1} + \frac{(K\gamma-t)(K+1)}{(t+1)(t+2)}
\end{align}
is achievable and it has a gap from optimal
\begin{align}
\frac{T''(K)}{T^*} \leq 4
\end{align}
that is less than 4.
\end{proposition}
\vspace{3pt}

The above maintains the gap from optimal of $4$, simply because the interpolation gives an improved performance over the case where $K\gamma \in [1,2,\dots,K]$ (see also~\cite{HA:2015}). The expression coincides with the original $T(K)$ for integer values of $K\gamma$. The purpose of this proposition is to allow for the applicability of Theorem~\ref{thm:centralT} without the integer relaxation assumption. With $T''(L)$ in place, Theorem~\ref{thm:centralT} can apply, simply now with slightly different values for $\bar \tau_{thr}$ and $\tau_{thr}$, which though are more complicated and which do not offer any additional insight and are thus omitted.

Below we briefly describe the scheme.

\subsubsection{Proof of Proposition~\ref{prop:generalT}} \label{sec:schemegeneral}
Let $\Gamma =\frac{KM}{N} \in [t, t+1],$ for some $t=0,1,\cdots,K-1$. Let us start by splitting each file $W_n$ into two parts, where the first part $W^{(1)}_n$ has size $\big((t+1)-K\gamma\big) f$ and the second part $W^{(2)}_n$ has size $(K\gamma-t)f$. Split each cache $Z_k$ into two parts, $Z_{k,1}, Z_{k,2}$ such that $\frac{|Z_{k,1}|}{|Z_{k,2}|} = \frac{\big((t+1)-K\gamma\big)}{(K\gamma-t)}$. Focusing on the first part, apply the original MN algorithm, where now the library is $\{W^{(1)}_n\}_{n=1}^{N}$, the caches are $\{Z_{k,1}\}_{k=1}^K$, and caching is performed as though $K\gamma = t$, i.e., by splitting each half-file $W^{(1)}_n$ into $\binom{K}{t}$ equally-sized subfiles $W^{(1)}_{n,\tau},\tau\in \Psi_{t}$ (each subfile now has size $((t+1)-K\gamma) f / \binom{K}{t}$), and by filling the caches according to $Z_{k,1}=\{W^{(1)}_{n,\tau}\}_{n \in [N], \tau\in \Psi_{t},k \in \tau}$. Then simply create the sequence of $\binom{K}{t+1}$ XORs (where now each XOR is intended for $t+1$ users), the delivery of which requires
\begin{align}
T^{(1)} = (t+1-K\gamma) \frac{\binom{K}{t+1}}{\binom{K}{t}}.
\end{align}

We then do the same for the second half of the files (second library $\{W^{(2)}_n\}_{n=1}^{N}$) except that now we substitute $t$ with $t+1$, to get a corresponding duration of
\begin{align}
T^{(2)} = (K\gamma-t) \frac{\binom{K}{t+2}}{\binom{K}{t+1}}.
\end{align}
Combining the two cases yields the whole duration of the delivery phase to be
\begin{align}
T =T^{(1)}+T^{(2)} = \frac{K-t}{t+1} + \frac{(K\gamma-t)(K+1)}{(t+1)(t+2)}
\end{align}
which completes the proof.

\bibliographystyle{IEEEtran}
\bibliography{IEEEabrv,final_refs}

\end{document}